\documentclass[manuscript,nonacm]{acmart}

\usepackage{graphicx}
\usepackage{amsmath}
\usepackage{xspace}
\usepackage{tikz, tikz-3dplot}
\usepackage{color, colortbl}
\usepackage{amsthm}
\usepackage{algorithm}
\usepackage[noend]{algpseudocode}
\usepackage{units}
\usepackage{thmtools}
\usepackage{thm-restate}
\usepackage{cleveref}

\newcommand{\RR}{\mathbb{R}}
\newcommand{\factor}{\frac{1}{2}}
\newcommand{\logbase}{2}

\newcommand{\Elim}[2]{\ensuremath{\mathsf{Elim}^{#1}(#2)}}
\newcommand{\furthest}[1]{\ensuremath{\mathsf{furthest}(#1)}}
\newcommand{\processMessages}{\ensuremath{\mathsf{process\_messages}()}}
\newcommand{\initNode}{\ensuremath{\mathsf{init\_node}()}}
\newcommand{\valueMessageReady}[1]{\ensuremath{\mathsf{value\_message\_ready}(#1)}}
\newcommand{\reportMessageReady}[1]{\ensuremath{\mathsf{report\_message\_ready}(#1)}}
\newcommand{\exVal}[1]{\ensuremath{\mathsf{ex\_val}(#1)}}
\newcommand{\vote}[1]{\ensuremath{\mathsf{vote}\left(#1\right)}}
\newcommand{\conv}[1]{\ensuremath{\mathsf{conv}\left(#1\right)}}
\newcommand{\diam}[1]{\ensuremath{\mathsf{diam}\left(#1\right)}}
\newcommand{\dist}[2]{\ensuremath{\mathsf{dist}\left(#1,#2\right)}}
\newcommand{\values}[2]{\ensuremath{values^{#1}_{#2}}}
\newcommand{\reports}[2]{\ensuremath{reports^{#1}_{#2}}}
\newcommand{\protocolName}[0]{\ensuremath{\mathsf{MultiDimApproxAgree}}}

\newcommand{\dottedcolumn}[3]{%
  \settowidth{\dimen0}{$#1$}
  \settowidth{\dimen2}{$#2$}
  \ifdim\dimen2>\dimen0 \dimen0=\dimen2 \fi
  \begin{pmatrix}\,
    \vcenter{
      \kern.6ex
      \vbox to \dimexpr#1\normalbaselineskip-1.2ex{
        \hbox{$#2$}
    \kern3pt
    \xleaders\vbox{\hbox to \dimen0{\hss.\hss}\vskip4pt}\vfill
    \kern1pt
    \hbox{$#3$}
  }\kern.6ex}\,
  \end{pmatrix}
}

\algdef{SN}{MainLoop}{EndMainLoop}{\textbf{repeat}}
\algdef{SE}{DelayUntil}{EndDelayUntil}[2]{\textbf{Upon} receiving #1 from node #2 \textbf{delay the message until:}}{\textbf{Once these conditions hold: }}
\algdef{SN}{UponIf}{EndUponIf}[2]{\textbf{upon} #1, \textbf{if} #2, \textbf{do}}
\algdef{SN}{Upon}{EndUpon}[1]{\textbf{upon} #1, \textbf{do}}
\algdef{SN}{MultiIf}{EndMultiIf}{\textbf{if} the following conditions hold }
\algdef{xCN}{MultiIf}{MultiThen}{EndMultiIf}{\textbf{then}}
\algdef{xCN}{MultiThen}{MultiElse}{EndMultiIf}{\textbf{else}}

\begin{document}





\author{Maya Dotan}
\author{Gilad Stern}
\author{Aviv Zohar}
\affiliation{\institution{Hebrew University}\country{Israel}}


\title{Validated Byzantine Asynchronous Multidimensional Approximate Agreement}

\begin{abstract}
Consider an asynchronous system where each node begins with some point in $\RR^m$. Given some fixed $\epsilon > 0$, we wish to have every nonfaulty node eventually output a point in $\RR^m$, where all outputs are within distance $\epsilon$ of each other, and are within the convex hull of the original nonfaulty inputs. This problem, when some of the nodes are adversarial, is known as the ``Byzantine Asynchronous Multidimensional Approximate Agreement'' problem.  

Previous landmark work by Mendes \textit{et al.} and Vaidya \textit{et al.} presented two solutions to the problem. Both of these solutions require exponential computation by each node in each round.
Furthermore, the work provides a lower bound showing that it is impossible to solve the task of approximate agreement if $n\leq (m+2)t$, and thus the protocols assume that $n>(m+2)t$.

We present a Byzantine Asynchronous Multidimensional Approximate Agreement protocol in the \emph{validated} setting of Cachin \textit{et al.}
Our protocol terminates after a logarithmic number of rounds, and requires only polynomial computation in each round.
Furthermore, it is resilient to $t<\frac{n}{3}$ Byzantine nodes, which we prove to be optimal in the validated setting.
In other words, working on the task in the validated setting allows us to significantly improve on previous works in several significant metrics.
In addition, the techniques presented in this paper can easily yield a protocol in the original non-validated setting which requires exponential computation only in the first round, and polynomial computation in every subsequent round.
\end{abstract}


\maketitle
\newpage

\section{Introduction}

One of the elementary problems in the field of distributed computing is the consensus problem.
A natural alternative to the agreement property in an asynchronous network is a weaker \emph{approximate} agreement property, which bypasses the famed FLP \cite{FLP1985} impossibility result.
In this setting, all nonfaulty nodes are required to output values that are no more than $\epsilon$ apart, for some known $\epsilon$.
In order to avoid trivial solutions, all nonfaulty nodes must output some value in the convex hull of the inputs they received. 
Requiring outputs to be in the convex hull of inputs also yields desirable properties.
For example, several parties might want to agree on the value of some given commodity, i.e. output some number that represents its value.
If the parties were to use a regular consensus protocol, while starting with different evaluations for the commodity's value, they would be allowed to output any agreed upon value from the protocol.
This is true, even if they originally had very similar evaluations for the commodity's value.
Outputting a value in the convex hull of their inputs guarantees that their output remains ``meaningful'' by making sure that it is between the highest and lowest input values.
This notion can also be generalized to agreeing on vectors in higher dimensions.
For example, parties might want to allocate resources between several different avenues.
If that is the case, again we would naturally want any agreed upon allocation to be ``in between'' the originally suggested allocations, i.e. in its convex hull.
Furthermore, we might also want to impose an additional validity restriction stating that the agreed upon value is indeed a full allocation of the resources, giving all of the resources, but not more than the original amount.

In 1986, Dolev \textit{et al.} \cite{dolev1986synchronous}
presented an optimally resilient approximate agreement protocol in the presence of Byzantine faults in \emph{synchronous} networks.
In 2004, Abraham~\textit{et al.} \cite{abraham2004optimal} constructed an optimally resilient Approximate Agreement protocol in the presence of Byzantine faults in \emph{asynchronous} networks for scalar inputs.
A natural generalization is Approximate Agreement protocols for inputs that can be multidimensional, i.e. in $\RR^m$ for some value $m > 1$.
Mendes \textit{et al.} \cite{mendes2015combined} constructed two protocols solving Byzantine Asynchronous Approximate Agreement, and showed that no such protocol can exist if the number of Byzantine nodes is $t$ for $n\leq t\left(m+2\right)$. 
They generalize the steps shown in \cite{abraham2004optimal} to higher dimension in order to show a possibility result under the aforementioned constraint. 
They break each round in the algorithm into a step of data collection, including a ``witness technique'', and a step of calculating the input for the next round.
This general framework proves very useful. Our work uses similar tools, however the calculation of the input for the next round is  different. 
We combine a simple geometric primitive with an enhanced witness technique that guarantees inter-round consistency of votes. This in turn reduces the overall computation from exponential to polynomial in every round, while maintaining the correctness of votes. 

Mendes \textit{et al.}'s work uses the notion of a ``safe area'' to compute votes in future rounds.
Their work provides important insights into the problem,  but unfortunately the calculation of a ``safe area'' requires exponential computation. 
In~\cite{fugger2018fast}, F{\"u}gger \textit{et al.} present a protocol for Byzantine Asynchronous Multidimensional Approximate Agreement, which makes the communication and round complexity independent of the dimension $m$. 
Their work changes the voting rule but uses the same notion of a safe area, and thus the  computation remains exponential. 
In all of these works, as the dimension grows, the assumptions on the fraction of faulty nodes becomes more and more stringent.
The exponential computation and the assumption on the fraction of faulty nodes severely limit the use of those protocols in real-world applications.

Our main contributions in this paper are twofold.
First of all, we present the first Byzantine Asynchronous Multidimensional Approximate Agreement protocol in the validated setting of Cachin \textit{et al.} \cite{cachin2001validity}, which requires only polynomial computation, requires only a logarithmic number of rounds and is resilient to up to $t$ Byzantine nodes such that $n\geq 3t+1$. 
Note that this ratio is constant, and does not increase with the dimension $m$, forgoing the strict requirements on the number of faulty nodes implied by the lower bound of Mendes \textit{et al.}.
Secondly, we discuss how to use the ideas and techniques presented in the paper in order to construct a protocol in the non-validated setting of Mendes \textit{et al.}, which only requires polynomial computation in the first round and is fully polynomial afterwards.
This new algorithm is conceptually simple and computable in polynomial time, requiring only the repeated deletion of the furthest points in a set and computing the average of a set of points.

The validated problem is nearly identical to Byzantine Asynchronous Multidimensional Approximate Agreement, except that it is set in the practical ``validated'' setting. 
In this setting there is an additional assumption that there exists some ``external validity'' function which allows for nodes to check if a value is ``valid'' as an input. 
The nodes are then required to output a value in the convex hull of the ``valid'' inputs.
The validated setting suggested in~\cite{cachin2001validity} is very natural for real-world applications. Consensus tasks have been extensively researched in this validated setting \cite{abraham2021reaching,abraham2019vaba,AbrahamS20,cachin2001validity,DUMBO20}, which has proved to be useful and interesting. 
For example, we can think of the consensus protocol as being run on servers which receive signed inputs from clients, and only correctly signed inputs are considered ``valid''.
In particular, this setting exists in blockchain systems, where the clients' signing rights can be seen as an external validity function.
Note that in this formulation there is no notion of a ``nonfaulty input'', only an input which can be validated.
Another type of an external validity function can make sure that the input ``makes sense''.
Such a function can be formulated from the example of resource allocation shown above.
As stated there, when agreeing on the allocation of resources, we want to make sure that they were fully allocated, and that the total allocated amount isn't greater than the total amount of resources.
A fitting external validity function in this case would be to allow only values within the $m+1$-dimensional simplex.
This validity function may be cryptographic or otherwise, and our solution uses the function as a black-box, regardless of the choice of implementation.

The $\vote{}$ algorithm is easy to understand and implement. 
While simple, it turns out to be strong enough to replace the prohibitively expensive safe area calculations in the validated setting.
Even more surprising is the the fact that it is possible to use only a single round of safe area calculations and use the $\vote{}$ algorithm in the rest of the protocol, using the original round in order to guarantee validity.

Our protocol requires a number of rounds logarithmic in the diameter of the set of valid inputs and $\frac{1}{\epsilon}$, and each round consists of a constant number of broadcasts by each node. 
The computation performed in each round is polynomial.
This means the message and round complexity is similar to those of the protocol shown in the scalar case \cite{abraham2004optimal}, and importantly the complexity does not increase with the dimension $m$.
The general structure of the protocol is extremely similar to the protocol of \cite{abraham2004optimal}, similarly to the protocols of \cite{mendes2013multidimensional,vaidya2013vectorconcensus,mendes2015combined}, but it uses different voting and round approximation rules.
Note that the voting and round approximation rules can be replaced with any other rules that guarantees the same properties by just adjusting the way nodes check consistency with previous rounds.

As remarked above, without the use of cryptography, one possible implementation of external validity is using only one exponential time round of Mendes \textit{et al.}'s safe area calculation \cite{mendes2015combined} which works in the traditional setting, and then continuing with our polynomial-computation protocol.
In this solution, the votes are guaranteed to be in the convex hull of nonfaulty inputs after the first round, and since our protocol validates inter-round consistency of the votes, they will continue being in the convex hull of the nonfaulty inputs throughout the rest of the protocol.
This means our protocol with slight adjustments can be used in the traditional setting, yielding a fully polynomial solution in all rounds except for the first (exponential) one.

\section{Model and Basic Definitions}

As stated above we will solve the Validate Byzantine Asynchronous Multidimensional $\epsilon$-Agreement problem.
In our model, there are $n$ nodes with point-to-point channels between every pair.
The nodes can send messages to each other, which are guaranteed to arrive, but can arrive after any finite time.
The adversary can control up to $t<\frac{n}{3}$ nodes, causing them to arbitrarily deviate from the protocol.
Classically, Byzantine Asynchronous $\epsilon$-Agreement is defined as follows:
\begin{definition}
A protocol for Byzantine Asynchronous $\epsilon$-Agreement over $\RR^m$ has the following properties:
\begin{itemize}
    \item \textbf{Termination.} All nonfaulty nodes  complete the protocol. 
    \item \textbf{Correctness.} Two nonfaulty nodes $i$ and $j$ that complete the protocol output $y_i,y_j\in\RR^m$ such that $\dist{y_i}{y_j}\leq \epsilon$.
    \item \textbf{Validity.} Let $x_i$ be node $i$'s input and let $I$ be the set of nonfaulty nodes. The output $y_j$ of a nonfaulty node $j$ is in $\conv{\left\{x_i|i\in I\right\}}$.
\end{itemize}
\end{definition}

Following ideas presented in \cite{cachin2001validity} and the model described in \cite{abraham2019vaba} we assume there exists some external function $\exVal$, which all nonfaulty nodes can run as a black box.
The purpose of the function is to determine whether a value is a valid input.
In general, we say that a value $v$ is a valid input if $\exVal{v}=true$.
We think of a setting in which all nodes, including the Byzantine ones, receive valid inputs in $\RR^m$.
Assuming a function $\exVal$, we define Validated Byzantine Asynchronous $\epsilon$-Agreement:
\begin{definition}
A protocol for Validated Byzantine Asynchronous $\epsilon$-Agreement over $\RR^m$ has the following properties:

\begin{itemize}
    \item \textbf{Termination.} Assume that nonfaulty nodes receive only valid inputs. Then all nonfaulty nodes complete the protocol. 
    \item \textbf{Correctness.} Two nonfaulty nodes $i$ and $j$ that complete the protocol output $y_i,y_j\in\RR^m$ such that $\dist{y_i}{y_j}\leq \epsilon$.
    \item \textbf{Validity.} The output $y_j$ of a nonfaulty node $j$ is in $\conv{\{v\in\RR^m|\exVal{v}= true\}}$.
\end{itemize}
\end{definition}

In exact Byzantine Agreement protocols, a Validity property which states that all nonfaulty nodes output some value $v$ such that $\exVal{v}=true$ still allows for trivial solutions, by checking all values and outputting some valid value deterministically.
However, exhaustively searching over all valid values in $\RR^m$ is impractical so the simple Validity property proposed above is enough to rule out trivial solutions in our setting.
Note that if the external validity function can be an arbitrary function $\mathsf{ex\_val}:\RR^m\to \{true,false\}$, then having black-box access to it does not allow for practical ways to find a predefined valid value.  

The first part of the paper will contain some basic definitions, and the analysis of the basic procedures we will use throughout this work. 
In broad strokes, throughout each round of the protocol each node aggregates the points it has seen, and calculates a new point called a ``vote'' to send in the next round.
A vote calculated using these procedures lies in the convex hull of the aggregated points which is crucial for solving the problem of Validated Approximate Agreement.
In the initialization round, nodes eliminate extreme points before calculating a vote for the next round of the protocol which lies in the convex hull of the remaining points.
Later in the section, we prove some useful properties of these algorithms.
We will later use these properties to show that the diameter of the set of nonfaulty votes shrinks exponentially fast.

\subsection{Basic Definitions and Algorithms}

\begin{definition}[Diameter]
Let $A$ be a finite set of points. We define the \textit{diameter} of $A$ as: 
$$\diam{A} := \max \{ \dist{x}{y} | x,y \in A \}\;.$$
Where the dist function is simply the Euclidean distance
\end{definition}

\begin{definition}[Convex Hull]
Let $A\in\RR^n$ be some set. The \textbf{convex hull} of $A$ is the set of all possible finite convex combinations of points in $A$. We denote the convex hull of $A$ by $\conv{A}$. So:
$$\conv{A} = \{ \sum_{i=1}^k \alpha_i a_i\mid k\in[\mathbb{N}] \wedge \sum_{i=1}^k \alpha_i = 1 \wedge \forall i\in[k] \textrm{ } a_i\in A \wedge \forall i\in[k] \textrm{ }  0\le \alpha_i \le 1\}\;.$$
\end{definition}

\begin{definition}\label{def::furthest}
    Let $A$ be a finite multiset of points in $\RR^m$. We denote: $\furthest{A}=argmax_{x,y\in A}\left\{\left\|x-y\right\|\right\}$ 
    I.e. $furthest\left(A\right)$ is the maximal distance pair in $A$.
\end{definition}

Note that \furthest{A} is defined to be a pair of points with maximal distance in $A$.
We assume that there is a deterministic tie-breaking rule for ordering points by distance (for instance, the lexicographic tie-braking rule on the pair of points). Hence from now on when we say ``Maximal distance pair'' we mean maximal distance according to Euclidean distance and lexicographic tie-breaking, and $\furthest{A}$ is defined to use that same tie-breaking rule.

For a set of points $A$ we define the procedures \Elim{t}{A} and \vote{A}. The first, \Elim{t}{A} takes a set of points $A$ in $\RR^n$ and iteratively removes the two maximal-distance pair of points in $A$ $t$ times.
This can be done by computing the distance between each pair of points, sorting them, and then going over the sorted list and repeatedly removing pairs from which neither point has previously been deleted until $t$ pairs are deleted.
\Elim{t}{A}, presented in \cref{alg:elim}, is used in the initialization round in order to help nodes estimate the distance between points they have seen in this round and votes of all nodes in the next round. The second procedure, \vote{A}, which simply computes the average of all points in $A$.
In each round other than the initialization round nodes collect information from other nodes, and then compute their votes for next round using the \vote{} procedure presented in  \cref{alg:vote}.

\begin{algorithm}
\caption{$\Elim{t}{A}$}\label{alg:elim}
\hspace*{\algorithmicindent} \textbf{Input: A set of points $A$} \\
\hspace*{\algorithmicindent} \textbf{Output: A new set of points $A'$} 
\begin{algorithmic}[1]
\Procedure{$\Elim{t}{A}$}{}
\For{$i\gets 1,\ldots,t$}
    \State{$A\gets A\setminus \furthest{A}$}
\EndFor
\State \textbf{return} $A$
\EndProcedure
\end{algorithmic}
\end{algorithm}

\begin{algorithm}
\caption{\vote{A}}\label{alg:vote}
\hspace*{\algorithmicindent} \textbf{Input: A multiset of points $A$} \\
\hspace*{\algorithmicindent} \textbf{Output: A point $v$ in $\conv{A}$} 
\begin{algorithmic}[1]
\Procedure{\vote{A}}{}
\State \textbf{return}
$\frac{\sum_{v\in A}v}{\left|A\right|}$
\EndProcedure
\end{algorithmic}
\end{algorithm}



\subsection{Properties of the Vote Procedure}\label{sec::voteProperties}
As stated above, in each round nodes send each  other votes and check that they're consistent with votes from previous rounds.
Nodes only accept votes if they're found to be consistent.
In addition, in the initialization round of the protocol, nodes share information in order to approximate the distance between points in the next round.
Each node must be able to bound that distance only using the information it sees in the initialization round. 

In this subsection we show that the algorithms presented in the previous subsection uphold the following properties: (1) If non-faulty nodes have a ``large enough'' intersection in the initialization round, then they are able to bound the distance between any two points in the next round, and calculate the number of rounds that they need to run in order to guarantee validated approximate agreement. (2) When nodes check that votes have been computed according to \vote{}, the diameter of accepted values shrinks by a constant factor in each round, and (3) all non-faulty nodes always vote inside the convex hull of the initial validated set.

This implies that nodes only need to go through a logarithmic number of rounds (in the diameter of the original votes and $\frac{1}{\epsilon}$) to ensure that all of the votes are at most $\epsilon$-distance apart.
Intuitively, all nonfaulty nodes need to make sure that each vote was calculated based on a ``large enough'' set, and that every other vote is based on a set with a ``large enough'' intersection (the exact definitions of both conditions are stated in the relevant lemmas).
We show that if these two conditions hold, the desired multiplicative shrinking takes place.

In all of the lemmas and claims in this section we assume that $n\geq 3t+1$. 

The first lemma shows that non-faulty nodes (who validated inter-round consistency) always vote inside the convex hull of the nodes validated in the initialization round.
\begin{lemma}\label{lemma::convexhull}
    Let $P$ be a multiset of points in $\RR^m$ such that $\left|P\right|\leq n$.
    Let $U\subseteq P$ such that $\left|U\right|\geq n-t$.
    Then $\vote{U}\in \conv{P}$ and $\vote{\Elim{t}{U}}\in \conv{P}$. 
\end{lemma}

\begin{proof}
    Remember that $\vote{U}=\sum_{u\in U} \frac{1}{\left|U\right|}\cdot u\in \conv{U}$ and that $U\subseteq P$.
    In addition, $\Elim{t}{U}\subseteq U$, and thus $\conv{\Elim{t}{U}}\subseteq \conv{U}$.
    As before, $\vote{\Elim{t}{U}}\in \conv{\Elim{t}{U}}\subseteq \conv{U}$
    The claim follows.
\end{proof}

The following lemma is used in bounding the distance between votes after the initialization round.
In the lemma, $U,V$ are to be thought of as the votes that two nodes receive before completing the initialization round.
If a node completes the initialization round while having seen the votes in $V$, it will compute a vote inside the convex hull of $\Elim{t}{V}$ for the next round.
In this context, the lemma shows that a node receiving the votes in $U$ can bound the distance between any of the votes it has seen and the vote of other nodes in the first round, which lie within $\Elim{t}{}$ of the votes they have seen.
This also allows the node to bound the distance between any two nodes' votes in the next round.

\begin{lemma}\label{lemma::elimIntersection}
    Let $U,V$ be multisets of points in $\RR^m$ such that $\left|U\cup V\right|\leq n$ and $\left|U\cap V\right|\geq n-t$,
    then $\diam{U}\geq \diam{\Elim{t}{V}}$ and $U\cap\Elim{t}{V}\neq\emptyset$.
\end{lemma}
\begin{proof}
    First note that $\left|V\cap U\right|\geq n-t$ and thus $\left|U\right|,\left|V\right|\geq n-t\geq 2t+1$, meaning that $\Elim{t}{V}$ is well defined and is the result of removing $t$ pairs of points from $V$.
    For the first part of the lemma, for every $i\in[t]$, denote $(p_i,q_i)$ to be the pair deleted in the $i$'th iteration of $\Elim{t}{V}$. 
    For every such $i$, let $d_i=\dist{p_i}{q_i}$.
    First of all, $d_1\geq d_2\geq\ldots\geq d_t\geq \diam{\Elim{t}{V}}$, because in iteration $i$ the furthest distance pair is removed, leaving only points whose distance is no greater than $d_i$.
    If for any $i\in[t]$, $p_i,q_i\in U$, then $\diam{U}\geq \dist{p_i}{q_i}\geq \diam{\Elim{t}{V}}$.
    Otherwise, for every $i\in[t]$, either $p_i\notin U$ or $q_i\notin U$.
    This means that in total at least $t$ different points from $V\setminus U$ have been deleted from $V$ throughout $\Elim{t}{V}$.
    By assumption $n\geq \left|U\cup V\right|=\left|U\cap V\right|+\left|U\setminus V\right|+\left|V\setminus U\right|\geq n-t +\left|V\setminus U\right|$, and thus $t\geq \left|V\setminus U\right|$.
    In other words, all points from $V\setminus U$ have been deleted from $V$ throughout the process of $\Elim{t}{V}$, leaving only points from $V\cap U$.
    Therefore, in this case $\Elim{t}{V}\subseteq U\cap V\subseteq U$, and thus $\diam{\Elim{t}{V}}\leq \diam{U}$.
    For the second part of the lemma, by assumption $\left|U\cap V\right|\geq n-t\geq 2t+1$.
    Only $2t$ points are removed from $V$ during $\Elim{t}{V}$, so $U\cap\Elim{t}{V}\neq \emptyset$ as well.
\end{proof}

The next two lemmas are used to show that the diameter of all accepted values shrinks by a factor of $2$ in each round.
This means that a logarithmic number of rounds is required to reach the desired diameter.
\cref{lemma::addingVote} is a technical lemma only to be used in the next lemma and shows that voting according to the set $A$ or the set $A$ with $A$ added to it yield the same vote.
In \cref{lemma::closeVotes}, read $P$ as the set of all accepted votes in a given round, and $U,V$ as two sets of votes collected by two nodes.
In this context, the lemma shows that any two nodes calculating votes according to sets with a large enough intersection will have close votes in the next round.
In the protocol, nodes guarantee that any two accepted votes are computed with respect to sets with a large intersection, meaning that the diameter of all accepted votes shrinks by a constant factor in each round.

\begin{restatable}{lemma}{addingVote}\label{lemma::addingVote}
    Let $A$ be a multiset of points in $\RR^m$, let $l\in \mathbb{N}$ and let $v=\vote{A}$.
    Define $B$ to be the multiset $A$ with the value $v$ added $l$ times.
    Then $\vote{A}=\vote{B}$.
\end{restatable}

\begin{restatable}{lemma}{closeVotes}\label{lemma::closeVotes}
    Let $P$ be a multiset of points in $\RR^m$ such that $\left|P\right|\leq n$.
    Let $U,V\subseteq P$ such that $\left|U\cap V\right|\geq n-t$.
    Then $\dist{\vote{U}}{\vote{V}}\leq \frac{1}{2} \diam{P}$.
\end{restatable}
The proof of \cref{lemma::addingVote,lemma::closeVotes} is provided in the appendix.

\section{Validated Byzantine Asynchronous Multidimensional Approximate Agreement} \label{section::main}
In this section we take the algorithms from above and construct a Validated Byzantine Asynchronous Multidimensional Approximate Agreement protocol with them. 
We begin by formally defining the setting of the problem, and then follow with an explicit algorithm for the solution. 
The protocol follows the general framework of the AAD protocol described in \cite{abraham2004optimal}, which has also proven to be useful in Mendes \textit{et al.}'s work \cite{mendes2015combined}.
In this framework, the protocol is divided into 3 main conceptual parts.

The first part is an initialization and round-estimation protocol.
In this protocol, all nodes collect other nodes' initial data, estimate how many rounds the protocol must run in order to reach $\epsilon$-agreement, and then output a vote to suggest in the next vote.
In our protocol, the initialization and round-estimation takes place in Algorithm~\ref{alg:init}, which is called in line~\ref{line:init} of Algorithm~\ref{alg::mainProtocol}.
After initialization, all nodes participate in a loop consisting of the next two parts.
First of all, in the loop each node broadcasts its current vote, and collects votes using an idea called a witness technique.
The witness technique is a simple two-round protocol in which all nodes first broadcast votes, wait to receive $n-t$ votes from other nodes, and then broadcast a report of the $n-t$ votes they collected.
Once the nodes see that $n-t$ of the reports they received contain all of the values they received, they complete the witness technique.
The protocol guarantees that every pair of nonfaulty nodes that complete it have seen at least $n-t$ common values.

The main job of the loop in line~\ref{line:mainLoop} is to execute the witness technique in each round, with adjustments made to check that values are consistent with information from previous rounds.
Finally, after completing the witness technique in each round, every node computes its vote for next round using our voting rule and starts the next round.
In our protocol, this stage takes place in line~\ref{line:calcVote} of Algorithm~\ref{alg::mainProtocol}.
In the non-validated setting it is crucial that this voting rule outputs a point in the convex hull of the values received from nonfaulty nodes.
This means that if the faulty nodes send values outside the convex hull of the nonfaulty values, these values need to somehow be ignored.
Ideally, the voting rule is one such that if all nonfaulty nodes use it to calculate their votes, then the diameter of the convex hull of nonfaulty votes shrinks by a constant multiplicative factor in each round.
If that is the case, the exponential shrinking yields a logarithmic round requirement.

The protocol presented in this work slightly adjusts the general framework of the AAD protocol.
In the initialization round of the protocol, nodes check whether the suggested values are externally valid.
Afterwards, our protocol employs a witness technique similar to the one in AAD, with the added functionality that nodes also check that the values sent in a given round are consistent with values received in previous rounds.
More precisely, when sending a vote, nodes must also provide the set of values and reports which were used to compute the vote in the previous round.
Before accepting any such message, nodes check that each vote was computed correctly, and recursively check that the values and reports from previous rounds are correct.
This allows nodes to guarantee that all of the votes are consistent between rounds.
In addition, whereas in the previous protocol only nonfaulty nodes must compute votes based on sets with a large intersection, this adjustment requires faulty nodes to also do so, or have their votes rejected.
This process makes sure that only ``valid'' values are accepted in late rounds, and thus can be thought of as an external validity function for those rounds. 
This stronger witness technique could be used in additional settings, and we discuss this idea in Section~\ref{sec::conclusions}.
Finally, the voting rule in our protocol is the \vote{} algorithm described in previous sections.
The algorithm satisfies the desired property of multiplicative shrinkage between rounds as shown in Lemma~\ref{lemma::shrinkingDiameter}.
In the first round external validity is verified.
In later rounds inter-round consistency is verified, which guarantees that votes are within the convex hull of the previous round's votes.
Overall this means that votes are guaranteed to be inside the convex hull of the externally-valid values in every round.
Algorithm~\ref{alg::mainProtocol} describes the exact behavior of each node. 
In our solution we assume the existence of broadcast channels for each node.
Each node can send a broadcast accompanied by a tag.
\begin{definition}
    A broadcast channel has the following properties:
    \begin{itemize}
        \item \textbf{Validity.} A nonfaulty node $i$ receives a broadcast from a nonfaulty node $j$ with a given tag, if and only if $j$ sent that broadcast with that tag.
        \item \textbf{Liveness.} If some nonfaulty node $i$ receives a broadcast from node $j$ with a given tag, every nonfaulty node eventually receives that broadcast from $j$ with that tag.
        \item \textbf{Uniqueness.} If two nonfaulty nodes receive two messages $m$, $m'$ from the same node with the same tag, then $m=m'$. 
    \end{itemize}
\end{definition}
In our protocol, each message's tag is comprised of the current round number and the type of message.
The type of message (e.g. a ``value'' message) is the first element of each broadcast.
These channels can be simulated using the information theoretically secure Reliable Broadcast protocol described in \cite{bracha1987broadcast}.

\begin{algorithm}[ht]\caption{\protocolName($v^0, \epsilon$)}\label{alg::mainProtocol}
\hspace*{\algorithmicindent} \textbf{Code for node $i$} \\
\hspace*{\algorithmicindent} \textbf{Input: A value $v^0\in \RR^m$, precision $\epsilon$.} \\
\hspace*{\algorithmicindent} \textbf{Output: A value $v\in \RR^m$ in the convex hull of the valid inputs.} 

\begin{algorithmic}[1]
    \State \textbf{global} $r\gets 0$
    \State \textbf{global} $termination\_times\gets\emptyset$ \Comment{all of the initializations are of multisets}
    \State \textbf{global} $waiting\_values\gets\emptyset$ \Comment{value messages waiting to be processed}
    \State \textbf{global} $waiting\_reports\gets\emptyset$ \Comment{report messages waiting to be processed}
    \State \textbf{global} $\forall r \geq 0\  \values{r}{}\gets\emptyset$ \Comment{values accepted from other nodes}
    \State \textbf{global} $\forall r \geq 0\  \reports{r}{}\gets\emptyset$ \Comment{sets of values other nodes reported seeing}
    \State \textbf{global} $halt\gets \infty$ \Comment{number of rounds to execute} 
    \State continually run $process\_messages$ in the background
    \State \textbf{call} $\initNode$ \label{line:init}
    \State $r\gets 1$
    
    \State broadcast $\left(``value", v^1,\values{0}{},\reports{0}{},1\right)$
    \While{$r < halt$} \label{line:mainLoop}
        \For {$M = \left(``value",v^s,rec\_vals^{s-1},rec\_reps^{s-1},s\right) \in waiting\_val\_messages$}
            \If{$\valueMessageReady{M}$}\label{line::valueCheck}
                \State $\values{s}{}.add(v^s)$
                \Comment{do also for messages from round $s< r$}
                \If {$s = r$ and $\left|\values{r}{}\right| = n-t$}\label{line::enoughValues} \Comment{do only if the message is from round $r$}
                    \State broadcast $\left(``report",\values{r}{},r\right)$
                \EndIf
            \EndIf
        \EndFor
        \For {$M = \left(``report",rec\_vals^s,s\right) \in waiting\_reports$}
            \If{\reportMessageReady{M}}\label{line::reportCheck}
                \State $reports^{s}.add(rec\_vals^s)$
                \If {$s = r$ and $\left|reports^r\right| = n-t$}
                    \State $r\gets r+1$
                    \State $v^r\gets\vote{\values{r-1}{}}$\label{line:calcVote}
                    \State broadcast $\left(``value", v^r,\values{r-1}{},\reports{r-1}{},r\right)$
                \EndIf
            \EndIf
        \EndFor
    \EndWhile
    \State output $v^r$ and terminate
\end{algorithmic}
\end{algorithm}

\begin{algorithm}[ht]\caption{$\processMessages$}\label{alg::messages}

\begin{algorithmic}[1]
    \Upon{receiving $\left(``init\_value",v\right)$ from node $j$} 
        \If{$\exVal{v}=true$}
            \State $\values{0}{}.add(v)$
            \If{$\left|\values{0}{}\right|= n-t$}
                \State broadcast $\left(``report",\values{0}{},0\right)$
            \EndIf
        \EndIf
    \EndUpon
    \Upon{receiving $\left(``enough",e\right)$ from $j$}
        \State $termination\_times.add(e)$
        \If{$\left|termination\_times\right|\geq n-t$}
            \State set $halt$ to be the $t+1$'th smallest value in $termination\_times$
        \EndIf
    \EndUpon
    \Upon{receiving $M=\left(``value",v^{s},rec\_vals^{s-1},rec\_reps^{s-1},s\right)$ from $j$}
        \State $waiting\_values.add(M)$
    \EndUpon
    \Upon{receiving $M=\left(``report",rec\_vals^s,s\right)$ from $j$}
        \State $waiting\_reports.add(M)$
    \EndUpon
\end{algorithmic}
\end{algorithm}

\begin{algorithm}[ht]\caption{$\valueMessageReady{M=\left(``value",v^s,rec\_vals^{s-1},rec\_reps^{s-1},s\right)}$}\label{alg::valueCheck}
\begin{algorithmic}[1]
    \MultiIf 
        \State $r\geq s$, \Comment{relevant round reached}
        \State $\left|rec\_vals^{s-1}\right|\geq n-t$, \Comment{enough values were sent}
        \State $\left|rec\_reps^{s-1}\right|\geq n-t$, \Comment{enough reports were sent}
        \State $\forall vals\in rec\_reps^{s-1},\ vals\subseteq rec\_vals^{s-1}$, \Comment{reports are also in values}
        \State $rec\_vals^{s-1}\subseteq \values{s-1}{}$,  \Comment{$i$ accepted the sent values}
        \State $rec\_reps^{s-1}\subseteq \reports{s-1}{}$, \Comment{$i$ accepted the sent reports}
    \MultiThen
        \If{$s=1$ and $v^{s}_j\in conv (\Elim{t}{rec\_vals^{s-1}})$}
            \State return $true$
        \ElsIf{$s>1$ and $v^{s}_j= \vote{rec\_vals^{s-1}}$}
            \State return $true$
        \EndIf
    \EndMultiIf
    \State return $false$
\end{algorithmic}
\end{algorithm}

\begin{algorithm}[ht]\caption{$\reportMessageReady{M=\left(``report",rec\_vals^s,s\right)}$}
\begin{algorithmic}[1]
    \MultiIf 
        \State $r\geq s$, \Comment{relevant round reached}
        \State $\left|rec\_vals^{s}\right|\geq n-t$, \Comment{enough values were sent}
        \State $rec\_vals^{s}\subseteq \values{s}{}$. \Comment{$i$ accepted the sent values}
    \MultiThen
        \State return $true$
    \MultiElse
        \State return $false$
    \EndMultiIf
\end{algorithmic}
\end{algorithm}

\begin{algorithm}[ht]\caption{$\initNode$}\label{alg:init}
\begin{algorithmic}[1]
    \State broadcast $\left(``init\_value",v^0\right)$
    \While{$\left|termination\_times\right|< n-t$ and $\left|reports^0\right|< n-t$}
        \For {$M = \left(``report",rec\_vals^0,0\right) \in waiting\_reports$}
            \If{$\reportMessageReady{M}$}\label{line::initReportCheck}
                \State $reports^{0}.add(rec\_vals^0)$
                \If {$\left|reports^0\right| = n-t$}
                    \State  $v^1\gets\vote{\Elim{t}{\values{0}{}}}$
                    \State $enough\gets \left\lceil \log_{\logbase}\left(\nicefrac{3\diam{\values{0}{}}}{\epsilon}\right)\right\rceil+1$ \label{line::enoughComputation}
                    \State broadcast $\left(``enough", enough\right)$
                \EndIf
            \EndIf
        \EndFor
    \EndWhile
\end{algorithmic}
\end{algorithm}

We will start by setting some conventions which will help us discuss the protocol.

\begin{definition}
    We say that node $i$ is ``in round $s$'' if its local $r$ variable equals $s$.
    Furthermore, we say that node $i$ ``reached round $s$'' if at any point throughout the protocol its local $r$ variable equals $s$. Similarly, we say that node $i$ ``completed round $s$'' if at any point throughout the protocol its local $r$ variable is greater than $s$.
\end{definition}
    
\begin{definition}
    For every variable $x$ defined in the described algorithms, we define $x_i$ to be node $i$'s $x$ variable.
    For example, $termination\_times_i$ is node $i$'s $termination\_times$ set.
\end{definition}

From this point on, our main goal will be to prove this main theorem:
\begin{restatable}{theorem}{mainTheorem}
$\protocolName(\cdot, \epsilon)$ is a Validated Byzantine Asynchronous $\epsilon$-Agreement protocol over $\RR^m$ resilient to up to $t\leq\frac{n-1}{3}$ Byzantine nodes.
\end{restatable}

This algorithm is optimally resilient and a sketch for a proof of the lower bound can be found in section~\ref{sec::lowerBound} of the appendix. In order to prove our main theorem, we will prove several lemmas about the properties of the protocol.

\subsection{Liveness Lemmas}
The first set of lemmas shows that in each round some progress is made.
In broad strokes, the lemmas show that no nonfaulty node gets stuck: at any point it time it either eventually completes the round it is in, or completes the protocol.
The ideas in the protocol are very similar to the ones in Abraham \textit{et al.} \cite{abraham2004optimal}, with adjustments made for the externally valid setting.
Therefore, the proofs also follow similar logic, and are provided in the appendix.

\begin{restatable}[Eventual Viewpoint Consistency]{lemma}{setProgress}\label{lemma::setProgress}
    Let $i,j$ be two nonfaulty nodes that reach round $r$ of the protocol. Then either (1) from some point on, for every $s\leq r$, $\values{s}{j}\subseteq \values{s}{i}$ and $\reports{s}{j}\subseteq \reports{s}{i}$ \textbf{or} (2) $i$ eventually completes the protocol and terminates.
\end{restatable}

\begin{restatable}[Init Termination]
{lemma}{initTermination}\label{lemma::initTermination}
    If all nonfaulty nodes participate in the protocol and have valid inputs, then they all complete $init\_node$.
\end{restatable}

\begin{restatable}[Iterations Eventually Terminate]{lemma}{roundProgress}\label{lemma::roundProgress}
    If all nonfaulty nodes participate in the protocol, reach round $r$ of the protocol for $r\geq 1$, and none of them terminate during it, then they all complete the round.
\end{restatable}

\subsection{Shrinking Diameter Lemmas}
The lemmas in this section show that the diameter of the set of all points that nonfaulty nodes receive and add to their respective \values{}{} sets shrinks by a constant multiplicative in each round.
Our strategy will be to use the results from Section~\ref{sec::voteProperties}, and show that the local sets that all nonfaulty nodes meet the conditions of the claims we proved there, which allows us to use \cref{lemma::closeVotes}.
The first lemma shows that using the witness technique guarantees that nodes only accept votes which were computed according to sets with a large intersection.

\begin{lemma}[Reported Viewpoints of Accepted Values Are Close]\label{lemma::valueIntersection}
    Let $i,j$ be two nonfaulty nodes and observe some round $r\geq 1$.
    If node $i$ added $v_p$ to $\values{r}{i}$ as a result of receiving $\left(``value",v_p,rec\_vals_p,rec\_reps_p,r\right)$ from $p$, and node $j$ added $v_q$ to $\values{r}{j}$ as a result of receiving  $\left(``value",v_q,rec\_vals_q,rec\_reps_q,r\right)$ from $q$, then $\left|rec\_vals_p\cap rec\_vals_q\right|\geq n-t$.
\end{lemma}

\begin{proof}
    Nodes $i$ and $j$ added $v_p$ and $v_q$ to $\values{r}{i}$ and $\values{r}{j}$ respectively, so they received 
    
    $(``value",v_p,rec\_vals_p,rec\_reps_p,r)$ and $\left(``value",v_q,rec\_vals_q,rec\_reps_q,r\right)$ messages respectively.
    Since $i,j$ added those points we know that eventually
    $rec\_reps_p\subseteq \reports{r-1}{i}$ and $rec\_reps_q\subseteq \reports{r-1}{j}$ and that $\left|rec\_reps_p\right|\geq n-t$, $\left|rec\_reps_q\right|\geq n-t$.
    Every element $S\in \reports{r-1}{i}$ was added as a result of receiving a ``report'' message from a unique node $l$ for that round, after checking that $\left|S\right|\geq n-t$.
    Since no node can send two ``report'' messages per round, this means that the total number of values in $rec\_reps_p\cup rec\_reps_q$ is no greater than $n$.
    Now, $2\left(n-t\right)\geq n+t+1$, so $rec\_reps_p\cap rec\_reps_q\neq\emptyset$.
    Observe some set $S\in rec\_reps_p\cap rec\_reps_q$.
    Seeing as $i,j$ added $v_p$ and $v_q$ to their \values{}{} sets, it must also be the case that
    $S\subseteq rec\_vals_p \subseteq \values{r-1}{i}$ and that  $S\subseteq rec\_vals_q \subseteq \values{r-1}{j}$.
    As shown above, every such $S$ is of size at least $n-t$, so $\left|rec\_vals_p\cap rec\_vals_q\right|\geq n-t$, because they both contain the set $S$. 
\end{proof}

\begin{lemma}[Shrinking Diameter]\label{lemma::shrinkingDiameter}
    Let $r\geq 1$ be some round number, and let $G$ be the set of all nonfaulty nodes.
    Observe the sets \values{r}{i} after the nonfaulty nodes stop adding elements to them.\footnote{Each \values{}{} set can contain no more than $n$ elements, so such a point much exist.}
    Also, observe the set $\values{r}{}=\bigcup_{i\in G} \values{r}{i}$, consisting of all points the nonfaulty nodes ever accept in round $r$. Then:
    $$\diam{\values{r+1}{}}\leq \factor \diam{\values{r}{}}\;.$$
\end{lemma}

\begin{proof}
    Observe two points $v_p,v_q\in \values{r+1}{}$.
    Assume that $v_p\in \values{r+1}{i}$ and $v_q\in \values{r+1}{j}$ for some $i,j\in G$.
    Node $i$ added $v_p$ to $\values{r+1}{i}$ after receiving the message $(``value",v_p,rec\_vals_p,rec\_reps_p,r+1)$ from $p$.
    Similarly, node $j$ added $v_q$ to $\values{r+1}{j}$ after receiving the message 
    $(``value",v_q,rec\_vals_q,rec\_reps_q,r+1)$ from $q$.
    Seeing as $i,j$ added the values $v_p,v_q$ to their respective \values{}{} sets, they found that $\left|rec\_vals_p\right|,\left|rec\_vals_q\right|\geq n-t$ and $v_p=\vote{rec\_vals_p}$,$v_q=\vote{rec\_vals_q}$.
    Furthermore, $rec\_vals_p\subseteq \values{r}{i}$, $rec\_vals_q\subseteq \values{r}{j}$ and thus $rec\_vals_p, rec\_vals_q\subseteq \values{r}{}$.
    From Lemma~\ref{lemma::valueIntersection}, we know that  $\left|rec\_vals_p\cap rec\_vals_q\right| \geq n-t$.
    Combining all of those observations we can see that the conditions of Lemma~\ref{lemma::closeVotes} hold. 
    Therefore:
    \begin{align}
        \dist{v_p}{v_q}&=
        \dist{\vote{rec\_vals_p}}{\vote{rec\_vals_q}}\\ 
        &\leq \factor \diam{\values{r}{}}\;.
    \end{align}
    Since we dealt with two arbitrary points $v_p,v_q\in \values{r+1}{}$, we conclude that for every $r\geq 1$:
    $$\diam{\values{r+1}{}}\leq \factor \diam{\values{r}{}}\;.$$
\end{proof}

\subsection{Initial Diameter Approximation Lemmas}
The lemmas in this section show that we can use the diameter of every nonfaulty node's \values{}{} set in the end of $init\_node$ to bound the distance between every two points accepted by any two nonfaulty nodes in round $1$.
This in turn means that any nonfaulty node's perceived diameter in round $0$ can be used to calculate the number of needed rounds.

\begin{lemma}[Initial Diameter Approximation]\label{lemma::initApproximation}
    Let $i,j,k$ be three nonfaulty nodes that started the first round of the protocol.
    Let $v_p\in \values{1}{i}$ and $v_q\in \values{1}{j}$.
    Then
    $\dist{v_p}{v_q}\leq 3\diam{\values{0}{k}}$ at the time $k$ completes $init\_node$.
\end{lemma}

\begin{proof}
    Since $v_p\in \values{1}{i}$, node $i$ received a
    $(``value",v_p,rec\_vals_p,rec\_reps_p,1)$
    broadcast from node $p$ such that eventually the message satisfies the conditions of Algorithm~\ref{alg::valueCheck}, menaing that $rec\_reps_p\subseteq \reports{0}{i}$ and that $\left|rec\_reps_p\right|\geq n-t$.
    Observe the sets $vals_k^0,reports_k^0$ at the time $k$ completes $init\_node$.
    At that time we know that both sets are at least of size $n-t$.
    Applying a counting argument similar to the one in Lemma~\ref{lemma::valueIntersection} we find that $rec\_reps_p\cap\reports{0}{k}\neq\emptyset$ and thus $\left|rec\_vals_p\cap vals_k^0\right|\geq n-t$.
    
    Nodes only add values to $vals^0$ after receiving an $``init\_value"$ broadcast.
    Nodes can only send one such broadcast, so the union of all such sets contains $n$ values at most.
    Note that $rec\_vals_p\subseteq vals_i^0, rec\_vals_q\subseteq vals_j^0$.
    Finally, $i$ and $j$ check that $v_p\in \conv{\Elim{t}{rec\_vals_p}}$ and $v_q\in \conv{\Elim{t}{rec\_vals_q}}$.
    Overall, $\left|vals_k^0\cap rec\_vals_p\right|\geq n-t$ and $\left|vals_k^0\cup rec\_vals_p\right|\geq n$. 
    From \cref{lemma::elimIntersection}, $\diam{vals_k^0}\geq \diam{\Elim{t}{rec\_vals_p}}$ and that $vals_k^0\cap \Elim{t}{rec\_vals_p}\neq \emptyset$.
    Similarly, $\diam{vals_k^0}\geq \diam{\Elim{t}{rec\_vals_q}}$ and $vals_k^0\cap \Elim{t}{rec\_vals_q}\neq \emptyset$.
    Set two points $x\in vals_k^0\cap \Elim{t}{rec\_vals_p}$ and $y\in vals_k^0\cap \Elim{t}{rec\_vals_q}$
    From the convexity of the distance function, $\dist{x}{v_p}\leq \diam{\Elim{t}{rec\_vals_p}}$ and $\dist{x}{v_q}\leq \diam{\Elim{t}{rec\_vals_q}}$.
    In addition, $x,y\in vals_k^0$, so $\dist{x}{y}\leq \diam{vals_k^0}$
    Combining these observations and using the triangle inequality, we get:
    
    \begin{align*}
        \dist{v_p}{v_q}&\leq \dist{v_p}{x} +\dist{x}{y} + \dist{y}{v_q}\\
        &\leq \diam{\Elim{t}{rec\_vals_p}} + \diam{vals_k^0} + \diam{\Elim{t}{rec\_vals_q}}\\
        &\leq \diam{vals_k^0} + \diam{vals_k^0} + \diam{vals_k^0}=3 \diam{vals_k^0}
    \end{align*}
\end{proof}

\subsection{Main Theorems}
\mainTheorem*
\begin{proof}
We will prove each of the properties holds.

\textbf{Termination.}
Assume all nonfaulty nodes participate in \protocolName and have valid inputs.
From Lemma~\ref{lemma::initTermination} we know that all of them complete the $init\_node$ call.
First we will show that at least one nonfaulty node terminates throughout the protocol.
Assume by way of contradiction that none of them do.
Using Lemma~\ref{lemma::roundProgress} and simple induction we can show that for every $r\in\mathbb{N}$, all nonfaulty nodes eventually reach round $r$.
Note that at the time $i$ completes $init\_node$, $\left|termination\_times\right|\geq n-t$ and thus $halt_i$ is well defined.
Since the $t+1$'th smallest element in $termination\_times$ is monotonically decreasing (adding elements cannot increase the $t+1$'th smallest values),
as $i$ continues participating in the protocol $halt_i$ can only become smaller.
That means that eventually, every nonfaulty node $i$ will see that $r\geq halt_i$ and terminate.

Now assume some nonfaulty node terminates.
Let $s$ be the smallest round number such that some nonfaulty node terminates in round $s$, and let $j$ be such a nonfaulty node.
Now assume by way of contradiction that some nonfaulty node $i$ never terminates.
Since no nonfaulty node terminates before round $s$, from Lemma~\ref{lemma::roundProgress} we know that all nonfaulty nodes eventually reach it.
Node $j$ terminated in round $s$, which means that at the time it terminated, it must have found that $s\geq halt_j$.
Eventually, $i$ will receive all of the $\left(``enough",e\right)$ messages that $j$ received, and therefore at that time $termination\_times_j\subseteq termination\_times_i$.
Since the $t+1$'th smallest element is monotonically decreasing, at that time $s\geq halt_j\geq halt_i$ and $i$ will terminate.

\textbf{Validity.}
Let $C=\conv{\left\{v\in\RR^m|\exVal{v}=true\right\}}$.
We will prove by induction that for every nonfaulty $i$ and $r\geq 0$, $\values{r}{i}\subseteq C$ and $v_i^{r+1}\in C$ .
For $r=0$, let $u\in \values{0}{i}$.
Node $i$ would have only added $u$ to $\values{0}{i}$ if $\exVal{u}=true$ and thus $u\in C$.
After completing $init\_node$, $v^1_i=\vote{\Elim{t}{\values{0}{i}}}$.
Note that $\Elim{t}{\values{0}{i}}\subseteq \values{0}{i}$.
From \cref{lemma::convexhull}, $\vote{\Elim{t}{\values{0}{i}}}\in\conv{\values{0}{i}}$.

Next, we will prove the case for $r=1$.
Observe any $v_j\in \values{1}{i}$.
Node $i$ would have only added the value to $\values{1}{i}$ if it received  $\left(``value",v_j,rec\_vals_j,rec\_reps_j,r\right)$ from $j$, and eventually found that $v_j\in\conv{\Elim{t}{rec\_vals_j}}$ and that $rec\_vals_j\subseteq \values{0}{i}$.
As shown above, $\values{0}{i}\subseteq C$, and since $\Elim{t}{}$ only removes points, $\Elim{t}{rec\_vals_j}\subseteq rec\_vals_j$. Therefore, $v_j\in\conv{\Elim{t}{rec\_vals_j}}\subseteq\conv{\values{0}{i}}\subseteq C$.
In addition, $i$ computes $v_i^{r+1}=\vote{\values{r}{i}}$.
From \cref{lemma::convexhull}, $v_i^{r+1}\in\conv{\values{r}{i}}\subseteq C$.

Now observe some $r>1$, and assume the claim holds for $r-1$.
Observe any $v_j\in \values{r}{i}$.
Since $i$ added $v_j$ to $\values{r}{i}$, it received  $\left(``value",v_j,rec\_vals_j,rec\_reps_j,r\right)$ from $j$, and eventually found that $rec\_vals_j\subseteq \values{r-1}{i}$ and $v_j=\vote{rec\_vals_j}$.
From \cref{lemma::convexhull}, $v_j\in\conv{rec\_vals_j}\subseteq\conv{\values{r-1}{i}}\subseteq C$.
Finally, $i$ computes $v_i^{r+1}=\vote{\values{r}{i}}$.
Again, following the same logic, we find that $v_i^{r+1}\in C$ as required.

If a nonfaulty node $i$ terminated, it must have first completed $init\_node$, computed $v^1_i$ and set $r=1$.
Throughout the protocol, $i$ only increments $r$, so when it completes the protocol, it outputs $v^r_i$ for $r\geq 1$.
As shown above, for every $r\geq 1$ $v^r_i\in C$, as required.

\textbf{Correctness.}
For every nonfaulty $i$ and round $r\geq 1$ we will observe the set $\values{r}{i}$ at a point in time where $i$ doesn't add any more values to it (i.e. at a time where it is maximal in size).
In addition, for every nonfaulty $i$ we will define $\values{init}{i}$ to be the set $\values{0}{i}$ at the time $i$ computes $enough_i$ in line~\ref{line::enoughComputation} of $init\_node$.
Define $G$ to be the set of nonfaulty nodes and for every $r\geq 1$ define $\values{r}{}=\bigcup_{i\in G} \values{r}{i}$.
First we would like to show that for every nonfaulty $i$, $3\diam{\values{init}{i}}\geq \diam{\values{1}{}}$.
Observe the values $v_p,v_q\in \values{1}{}$ such that $\dist{v_p}{v_q}=\diam{\values{1}{}}$.
Assume that $v_p\in \values{1}{j}$ and $v_q\in \values{1}{k}$ for some $j,k\in G$.
By Lemma~\ref{lemma::initApproximation} we know that $\diam{\values{1}{}}=\dist{v_p}{v_q}\leq 3\diam{\values{init}{i}}$.

We now turn to observe the sets $\values{r}{}$ for every $r\geq 1$.
Using Lemma~\ref{lemma::shrinkingDiameter}, we know that for every $r\geq 1$:
$$\diam{\values{r+1}{}}\leq\factor \diam{\values{r}{}}\;.$$
Combining the two previous observations, we now know that for every nonfaulty $i$ and round $r\geq 1$:

$$\factor^{\left(r-1\right)}\cdot 3\diam{\values{init}{i}}\geq
\factor^{\left(r-1\right)} \diam{\values{1}{}}\geq 
\diam{\values{r}{}}$$
Setting $r\geq enough_i$, we find that:
\begin{align*}
    \diam{\values{r}{}}&\leq
    3\diam{\values{init}{i}}\cdot \factor^{\left(r-1\right)}\\
    & \leq 3\diam{\values{init}{i}}\cdot \factor^{\left(\left\lceil \log_{\logbase}\left(\nicefrac{3\diam{\values{init}{i}}}{\epsilon}\right)\right\rceil+1-1\right)}\\
    & \leq 3\diam{\values{init}{i}}\cdot \factor^{\left( \log_{\logbase}\left(\nicefrac{3\diam{\values{init}{i}}}{\epsilon}\right)\right)} \\
    & = 3\diam{\values{init}{i}} \cdot \frac{\epsilon}{3\diam{\values{init}{i}}}=\epsilon\;.
\end{align*}

Finally, observe the smallest round $s$ in which some nonfaulty node terminates, and let $i$ be such a node.
First, clearly if all nonfaulty nodes would have sent their values in round $s$, they would all eventually accept those values and add them to their respective \values{s}{} sets.
In that case all of the arguments above would also apply to those values.
This means that we can consider all of the values $v^s_i$ output by nonfaulty nodes to be part of \values{s}{}.
Observe $halt_i$ at the time $i$ terminates.
As shown in the proof of the Validity property, we know that for some nonfaulty node $j$, $s\geq halt_i\geq enough_j$.
From our previous observation, we know that $\diam{\values{s}{}}\leq \epsilon$.
Furthermore, we've also shown in the proof of the Validity property that for every nonfaulty $i$ and round $r\geq 1$, $\values{r}{i}\subseteq \conv{\values{r-1}{i}}$, which means that also $\values{r}{}\subseteq \conv{\values{r-1}{}}$.
In other words all nonfaulty nodes output values in $conv\left(
\values{s}{}\right)$, and thus all of their outputs are no more than $\epsilon$ apart.

\end{proof}

\begin{theorem}
Let $V=\left\{v\in \RR^m|\exVal{v}\right\}$ be the set of all valid inputs.
Every nonfaulty node runs for $O\left(\log\left(\frac{ \diam{V}}{\epsilon}\right)\right)$ rounds.
\end{theorem}
\begin{proof}
    Observe some nonfaulty node $i$.
    At the time $i$ completes $init$, $\left|termination\_times_i\right|\geq 2t+1$.
    At that time $halt_i$ is defined to be the $t+1$'th smallest value in $termination\_times_i$.
    Since there are $n-t\geq 2t+1$ values in $termination\_times_i$, there are at least $t+1$ values $e\in termination\_times_i$ such that $halt_i\leq e$.
    There are only $t$ faulty nodes, so $halt_i\leq e$ for a value $e$ that some nonfaulty $j$ sent.
    Node $j$ sends the value $enough_j=\left\lceil \log_{\logbase}\left(\nicefrac{3\diam{\values{0}{j}}}{\epsilon}\right)\right\rceil+1$.
    As shown in the proofs of the Correctness and Validity properties of \protocolName, $halt_i$ can only become smaller throughout the algorithm, and $\values{0}{j}\subseteq V$.
    This leads us to conclude that $i$ terminates after seeing that $r\geq \left\lceil \log_{\logbase}\left(\nicefrac{3\diam{\values{0}{j}}}{\epsilon}\right)\right\rceil+1$, and thus won't run for more than $ \left\lceil \log_{\logbase}\left(\nicefrac{3\diam{V}}{\epsilon}\right)\right\rceil+1$ rounds.
\end{proof}

\section{Conclusions and Future Work}\label{sec::conclusions}
In this work we formalized the task of Validated Byzantine Asynchronous $\epsilon$-Agreement over $\RR^m$, provided an efficient protocol that solves it, and proved it is optimally resilient.
In future works we would like to further explore the task in the traditional setting (i.e. in which nodes output values in the convex hull of the nonfaulty inputs) and check if our method can help solve the problem efficiently.
We believe that our extended witness technique is also relevant to the traditional variant of the problem.
Using this technique, hybrid versions of the Mendes-Herlihy and Vaidya-Garg algorithms \cite{mendes2015combined} and our method can be used to derive more computationally efficient solutions by running one initialization round which requires exponential computation.
In this round, nodes will compute the Safe Area as described in \cite{mendes2015combined}.
Each output of that computation is in the convex hull of the nonfaulty inputs.
After the initialization round, nodes can check inter-round consistency using our protocol, which will prevent faulty nodes from reporting values outside of the convex hull of nonfaulty inputs.
By switching over to our validated solution, the protocol still requires a logarithmic number of rounds, with only the first one requiring exponential computation. 
In fact there may be many ways of running an initialization round, and we think this is an interesting future direction.

Another avenue of future research could involve also relaxing the validity property of the protocol.
In addition to only requiring $\epsilon$-Correctness, we could conceivably require only $\delta$-Validity. 
By this we mean that nodes aren't required to output values in the convex hull of nonfaulty inputs, but to output values that are close to nonfaulty values. 
One way to formalize this is by defining $V$ to be the set of all nonfaulty inputs, and allowing parties to output values in a ball of diameter $\delta\cdot \diam{V}$ which contains all points in $V$.
Note that this condition still rules out trivial solutions by requiring the ball to contain the values in $V$.
The $\protocolName$ protocol also achieves $3$-Validity in the traditional setting, i.e. when ignoring the external validity function (or setting it to always return $true$). 
This can be shown by proving the following informally stated lemma:
\begin{restatable}{lemma}{deltaValidity}
Let $V$ be the set of all nonfaulty inputs.
When executing $\protocolName$ with $\exVal{}$ always returning $true$, if a nonfaulty node accepts a value $v$ in round $1$, then there exists some nonfaulty input $x\in V$ such that $\dist{v}{x}\leq\diam{V}$.
\end{restatable}
The lemma is formally stated and proved in the appendix.
Now we know that in the first round of the protocol, all values accepted by all nonfaulty nodes are of distance $\diam{V}$ or less from the convex hull of nonfaulty inputs, which means that they are within a ball of diameter $3\diam{V}$ which contains $V$.
Following the proof in the rest of the paper, we can now conclude that every nonfaulty node's output will be within that ball\footnote{As currently stated, the adversary can choose very distant points in the initialization round, inflating the number of rounds required. In order to mitigate this, two initialization rounds can be used instead: the first resulting in the distance between accepted points being no more than $3\diam{V}$, and the second used to compute the number of required rounds by approximating the diameter of the resulting ball.}.
As far as we know, this is the first formulation of such a property and future research might lead to a much smaller $\delta$ by using different techniques in the initial rounds.


\bibliographystyle{plainurl}
\bibliography{bib}

\newpage
\section{Appendix}
\appendix

\section{Lower Bound Sketch }\label{sec::lowerBound}
It is natural to ask whether the presented protocol is optimally resilient, i.e. whether there could exist a Validated Byzantine Asynchronous Multidimensional $\epsilon$-Agreement protocol resilient to $t$ Byzantine nodes such that $n\leq 3t$.
It is important to state that \cite{mendes2015combined} proves a lower bound showing that there does not exist an approximate agreement protocol over $\RR^m$ resilient to any $t$ such that $n\leq \left(m+2\right)t$.
The protocol presented in this paper manages to circumvent this lower bound by assuming an external validity function and only requiring valid outputs to be within the convex hull of the valid inputs.
We can show that no Validated Byzantine Asynchronous Multidimensional $\epsilon$-Agreement protocol can be resilient to $t\geq\frac{n}{3}$ faulty nodes using ideas from \cite{DLS1988}.

Set some $t,n$ such that $n\leq 3t$. In addition, set some $\epsilon$ and $m$.
Assume by way of contradiction that there exists some Validated Byzantine Asynchronous $\epsilon$-Agreement protocol over $\RR^m$ resilient to $t$ faults.
Uniformly and independently sample two values $v_1\in\left(0,1\right)^m$ and $v_2\in\left(\epsilon+1,\epsilon+2\right)^m$ and define $\exVal$ as follows:
$$\exVal{x}=\begin{cases}
true\ &x\in\left\{v_1,v_2\right\} \\
false\ &else
\end{cases}\;.
$$
Let $B\subseteq \left[n\right]$ be some set of nodes of size $t$ that the adversary will control.
If $n-t=1$, the adversary will control only $t-1$ nodes instead.
Divide the remaining nodes into two sets $G_1$ and $G_2$ of sizes $\left\lceil\frac{n-t}{2}\right\rceil$ and $\left\lfloor\frac{n-t}{2}\right\rfloor$ respectively.
All nodes in $G_1$ will receive $v_1$ as input and all nodes in $G_2$ will receive $v_2$ as input and run the protocol.
All nodes in $B$ will communicate with nodes in $G_1$ as nonfaulty nodes would with input $v_1$.
Similarly all nodes in $B$ will communicate with nodes in $G_2$ as nonfaulty nodes with input $v_2$.
All messages between nodes in $B$ and $G_1$ and between nodes in $B$ and $G_2$ are delivered instantly.
On the other hand, all communication between nodes in $G_1$ and $G_2$ is delayed until all of those nodes complete the protocol.
Technically, this scheduling might not be valid if those nodes never complete the protocol, but we will show that this is not the case.
From the point of view of all nodes in $G_1$, any execution of the protocol in this setting indistinguishable from a setting where the adversary controls all nodes in $G_2$ and instructs them to be silent.
If the adversary controls the nodes in $G_2$ and the rest of the nodes communicate freely, the nodes in $G_1$ must complete the protocol.
Since those two settings are indistinguishable to nodes in $G_1$, they must do complete protocol if the adversary controls the nodes in $B$ as well.
The same argument can be made for nodes in $G_2$.
Now, from the point of view of nodes in $G_1$, this setting is indistinguishable from a setting in which the only valid input is $v_1$.
This is because the only other valid value was uniformly and independently sampled from $\left(\epsilon+1,\epsilon+2\right)^m$.
The probability that any given value is chosen is $0$, so there is no way for any node in $G_1$ to find out which other value was sampled (or even whether there exists any other value $v\in\RR^m$ such that $\exVal{v}=true$).
From the Validity property, every node in $G_1$ must then output the value $v_1$.
Using the exact same arguments, every node in $G_2$ must output the value $v_2$.
Note that from the way $v_1$ and $v_2$ were sampled we know that $\dist{v_1}{v_2}> \epsilon$, reaching a contradiction to the Correctness property of the protocol.

\section{Proofs for Technical Claims and Lemmas }\label{app::voteProofs}

\addingVote*
\begin{proof}
    Let $\left|A\right|=k$.
    The set $B$ is the multiset $A$ with $v$ added $l$ times, so $\left|B\right|=k+l$.
    Then:
    \begin{align*}
        \vote{B}&= \frac{1}{k+l}\sum_{u\in B} u\\
        &= \frac{1}{k+l} (l\cdot v+\sum_{u\in A}u)\\
        &= \frac{1}{k+l} (l(\frac{1}{k}\sum_{u\in A} u)+\frac{k}{k}(\sum_{u\in A}u))\\
        & = \frac{1}{l+k} (l+k)(\frac{1}{k}\sum_{u\in A} u)=\frac{1}{k}\sum_{u\in A} u=\vote{A}
    \end{align*}
\end{proof}

\closeVotes*
\begin{proof}
    Let $k=\left|U\setminus V\right|$ and $l=\left|V\setminus U\right|$.
    Note that:
    \begin{align*}
        n&\geq \left|U\cup V\right|=\left|U\cap V\right|+\left|U\setminus V\right|+\left|V\setminus U\right| \geq n-t + k +l
    \end{align*}
    and thus $t\geq k+l$.
    Define $U'$ to be the multiset $U$ with the value $\vote{U}$ added $l$ times and similarly define $V'$ to be the multiset $V$ with the value $\vote{V}$ added $k$ times.
    Note that now:
    \begin{align*}
        \left|U'\right|= \left|U\right|+l=\left|U\cap V\right|+\left|U\setminus V\right| +l= \left|U\cap V\right| + k + l
    \end{align*}
    and similarly, $\left|V'\right|=\left|U\cap V\right|+k+l$, and thus $\left|V'\right|=\left|U'\right|$.
    Also, $U\cap V\subseteq U\subseteq U'$, and similarly $U\cap V\subseteq V'$.
    This means that both $U'$ and $V'$ each consist of all elements in $U\cap V$, and each of the sets includes $k+l$ additional elements (some of which might be shared).
    Let $U'\setminus (U\cap V)=\{u_1,\ldots,u_{k+l}\}$ and $V'\setminus (U\cap V)=\{v_1,\ldots,v_{k+l}\}$.
    Note that for every set $A$, $\vote{A}\in \conv{A}$.
    By definition, for every $u\in U,v\in V$, $\dist{u}{v}\leq \diam{P}$.
    Distance is a convex function in each of its arguments, so for every $u\in U',v\in V'$ it is also true that $\dist{u}{v}\leq \diam{P}$.
    
    From \cref{lemma::addingVote}, $\vote{U}=\vote{U'}$ and $\vote{V}=\vote{V'}$, so $\dist{\vote{U}}{\vote{V}}=\dist{\vote{U'}}{\vote{V'}}$.
    Combining the previous observations:
    \begin{align*}
        \dist{\vote{U'}}{\vote{V'}} &= \left\lVert \vote{U'}-\vote{V'}\right\rVert\\
        &= \left\lVert \frac{1}{\left|U'\right|}\sum_{u\in U'} u-\frac{1}{\left|V'\right|}\sum_{v\in V'} v\right\rVert\\
        &=\frac{1}{\left|U'\right|}\left\lVert\sum_{u\in U'\setminus (U\cap V)} u-\sum_{v\in V'\setminus (U\cap V)} v\right\rVert\\
        & \leq \frac{1}{n-t}\left\lVert\sum_{u\in U'\setminus (U\cap V)} u-\sum_{v\in V'\setminus (U\cap V)} v\right\rVert\\
        &= \frac{1}{n-t}\left\lVert\sum_{i=1}^{k+l} u_i-\sum_{i=1}^{k+l} v_i\right\rVert\\
        &\leq \frac{1}{n-t}\sum_{i=1}^{k+l}\left\lVert u_i- v_i\right\rVert\\
        &\leq \frac{1}{n-t}(k+l)diam (P)\\
        &< \frac{1}{2t}\cdot t\cdot \diam{P}= \frac{1}{2}\diam{P}
    \end{align*}
\end{proof}

\setProgress*
\begin{proof}
    If $i$ completes the protocol, we are done. If that is not the case, we will prove the claim holds by induction on $s$.
    For $s=0$, if $v_k\in \values{0}{j}$, then $j$ received an $\left(``init\_value",v_k\right)$ broadcast from some node $k$ s.t. $\exVal{v_k}=true$.
    From the Liveness and Uniqueness properties of the broadcast channel, $i$ will also receive that message from node $k$, find that $\exVal{v_k}=true$ and add $v_k$ to $\values{0}{i}$.
    From the Uniqueness property of the broadcast channel there can only be $n$ such broadcasts, so the number of elements in $value^0_{j}$ is bounded, and we can make this argument for very such element.
    In addition, if $rec\_vals_k \in \reports{0}{j}$, then $j$ must have received a $\left(``report",rec\_vals_k,0\right)$ broadcast from $k$ and at some point in time $rec\_vals_k\subseteq \values{0}{j}$, $\left|rec\_vals_k\right|\geq n-t$.
    From the Liveness and Uniqueness properties of the broadcast channel, $i$ will also receive the $\left(``report",rec\_vals_k,0\right)$ message from $k$.
    Seeing as the sets only grow, following the same logic as above eventually $rec\_vals_k\subseteq \values{0}{j} \subseteq \values{0}{i}$ and $\left|rec\_vals_k\right|\geq n-t$, at which point $i$ will add $rec\_vals_k$ to $\values{0}{i}$.
    
    Now assume the claim holds for some $s-1$.
    Observe some $v_k\in \values{s}{j}$. Since $j$ is nonfaulty, it must have received a $\left(``value",v^s_k,rec\_vals^{s-1}_{k},rec\_reps^{s-1}_{k},s\right)$ broadcast and added it to $waiting\_reports_j$.
    Then, at some point in time, $j$ found in line~\ref{line::valueCheck} that the message satisfies the conditions of Algorithm~\ref{alg::valueCheck}.
    From the Liveness and Uniqueness properties of the broadcast protocol, $i$ will receive that message as well and add it to $waiting\_values_i$. 
    By assumption $i$ reaches round $r$, and thus also reaches round $s\leq r$. 
    By the induction hypothesis $\values{s-1}{j}\subseteq \values{s-1}{i}$ as well as $\reports{s-1}{j}\subseteq \reports{s-1}{i}$.
    In other words, eventually $i$ will see that the conditions of Algorithm~\ref{alg::valueCheck} hold in line~\ref{line::valueCheck} and add $v^s_k$ to $\values{s}{i}$.
    Similarly, if $rec\_vals_k\in reports^s_{j}$, then $j$ received a message $\left(``report",rec\_vals_k,s\right)$ and added it to $waiting\_reports_j$.
    Then, $j$ found in line~\ref{line::reportCheck} that $rec\_vals_k\subseteq \values{s}{j}$, $\left|rec\_vals_k\right| \geq n-t$.
    From the Liveness and Uniqueness properties of the broadcast channel, eventually node $i$ will receive that broadcast from $k$ as well and add it to $waiting\_reports_i$
    We've already shown that eventually $i$ will reach round $s$ and $rec\_vals_k\subseteq \values{s}{j}\subseteq \values{s}{i}$, $\left|rec\_vals_k\right|\geq n-t$.
    At that point $i$ will add $rec\_vals_k$ to $\reports{s}{i}$, completing our proof.
\end{proof}

\initTermination*
\begin{proof}
    Assume all nonfaulty nodes participate in the protocol and have valid inputs.
    In the beginning of the protocol, they call $init\_node$ and broadcast an $\left(``init\_value",v^0_i\right)$ message.
    From the Validity property of broadcast channels, every nonfaulty node $i$ will receive those messages.
    Since all of them have valid inputs, $i$ will find that $\exVal{v^0_j}=true$ for every nonfaulty $j$, and add $v^0_j$ to $\values{0}{i}$.
    After that, every nonfaulty $i$ will find that $\left|\values{0}{i}\right|\geq n-t$ and send a $\left(``report",\values{0}{i},0\right)$ message.
    Note that every nonfaulty node $j$ that completes $init\_node$ has at least one value in its $reports^0_j$ set.
    A nonfaulty $j$ will only add the set $rec\_vals$ to $reports^0_j$ if $\left|rec\_vals\right|\geq n-t$ and $rec\_vals\subseteq \values{0}{j}$.
    This implies that $\left|\values{0}{j}\right|\geq n-t$, which means that $j$ must have sent some ``report'' message for round 0.
    In other words, every nonfaulty $j$ that completes $init\_node$ also broadcast some ``report'' message for round 0.
    
    Now observe some nonfaulty $i$.
    If $i$ completes the protocol, then it must have clearly completed $init\_node$ as well.
    Otherwise, it received a $\left(``report",rec\_vals^0_j,0\right)$ message from every nonfaulty $j$ such that $rec\_vals^0_j\subseteq \values{0}{j}$ and $\left|rec\_vals^0_j\right|\geq n-t$.
    We know from Lemma~\ref{lemma::setProgress} that eventually for every nonfaulty $j$, $\values{0}{j} \subseteq \values{0}{i}$, at which point $i$ will find that the conditions in line~\ref{line::initReportCheck} of $init\_node$ hold and add $rec\_vals^0_j$ to $reports^0_{i}$.
    After adding such a tuple for every nonfaulty node, $i$ will find that $\left|reports^0_{i}\right|= n-t$ and broadcast an ``enough'' message.
    From the Validity property of broadcast channels, every nonfaulty node $i$ will receive the ``enough'' message sent by all nonfaulty nodes.
    At that point, both $\left|termination\_times\right|\geq n-t$ and $\left|reports^0_i\right|\geq n-t$, and $i$ will complete $init\_node$.
\end{proof}

\roundProgress*
\begin{proof}
    In the beginning of the round, every nonfaulty node broadcasts the message $(``value",\\ v^r_i, \values{r-1}{i}, \reports{r-1}{i},r)$.
    From the Validity property of broadcast channels, every nonfaulty node will receive the ``value'' message sent by every nonfaulty node in round $r$.
    Observe some nonfaulty node $i$.
    If node $i$ never completes round $r$, it must never terminate because $i$ doesn't terminate during the round.
    In that case, we know from Lemma~\ref{lemma::setProgress} that from some point on $\values{r-1}{j}\subseteq \values{r-1}{i}$ and $reports^{r-1}_{j}\subseteq reports^{r-1}_{i}$ for every nonfaulty node $j$.
    Therefore, if $j$ broadcasts some message $\left(``value", v^r_j, rec\_vals^{r-1}_{j}, rec\_reps^{r-1}_{j},r\right)$, then we know that eventually $rec\_vals^{r-1}_{j} \subseteq \values{r-1}{j}\subseteq \values{r-1}{i}$ and $rec\_reps^{r-1}_{j} \subseteq \reports{r-1}{j}\subseteq \reports{r-1}{i}$.
    Note that $j$ proceeds to round $r$ only after $\left|\reports{r-1}{j}\right|\geq n-t$ and at that point it computes $v^r_j=\vote{ \values{r-1}{j}}$.
    Furthermore, $j$ would have only added a set $rec\_vals$ to $\reports{r-1}{j}$ if $rec\_vals\subseteq \values{r-1}{j}$ and $\left|rec\_vals\right|\geq n-t$.
    This also implies that $\left|\values{r-1}{j}\right|\geq n-t$.
    Combining these observations, eventually $i$ will see that the conditions in line~\ref{line::valueCheck} hold, and add $v^r_j$ to $\values{r}{i}$.
    After adding such a value for every nonfaulty node, $i$ will find that $\left|\values{r}{i}\right|\geq n-t$, and broadcast a $\left(``report",\values{r}{i},r\right)$ message.
    It is important to note that if $i$ completes round $r$, then $\left|\reports{r}{i}\right|\geq n-t$.
    As shown above, this must mean that $\left|\values{r}{i}\right|\geq n-t$ and thus $i$ sends a ``report'' message for round $r$.
    
    Observe some nonfaulty node $i$ again.
    If $i$ completes the protocol, then it must have completed round $r$ first.
    Otherwise, $i$ never completes the protocol and therefore receives a $\left(``report",rec\_vals^r_{j},r\right)$ message from every honest node $j$ and adds it to $waiting\_reports$.
    Following similar arguments to the ones above we know that eventually $rec\_vals^r_{j}\subseteq \values{r}{j} \subseteq \values{r}{i}$.
    Furthermore, $j$ only sends a ``report'' message after finding the conditions of line~\ref{line::enoughValues} hold, at which point $\left|\values{r}{j}\right|\geq n-t$.
    This means that for every honest $j$, node $i$ will eventually see that the conditions in line~\ref{line::reportCheck} hold, and add $rec\_vals^r_{j}$ to $\reports{r}{i}$.
    After adding a tuple for each honest node, $i$ sees that $\left|\reports{r}{i}\right|\geq n-t$ and continue to the next line.
    Afterwards $i$ will perform a few local computations, and start the next round. 
\end{proof}

\deltaValidity*
\begin{proof}
    First of all we will more formally state the lemma.
    Define $\exVal{x}$ to be true for every $x\in\RR^m$.  
    Let $G$ be the set of all nonfaulty nodes and let $V$ be the set of all nonfaulty inputs, i.e. $V=\{x_i|i\in G\}$.
    Similarly to the proof of the Correctness property of the protocol, for every nonfaulty $i$ observe $\values{1}{i}$ at a point in time where $i$ doesn't add any more values to it.
    Define $\values{1}{}=\bigcup_{i\in G} \values{1}{i}$.
    Then, for every $v\in\values{1}{}$ there exists some $x\in V$ such that $\dist{v}{x}\leq \diam{V}$.
    
    Observe some $v\in\values{1}{}$.
    Some nonfaulty node $i$ must have added $v$ to its $\values{1}{i}$ set, so it received a
    
    $\left(``value",v,rec\_vals,rec\_reps,1\right)$ broadcast such that $rec\_vals\subseteq \values{0}{}$, $\left|rec\_vals\right|\geq n-t$, 
    
    and $v\in\conv{\Elim{t}{rec\_vals}}$.
    The set $\values{0}{i}$ only contains values which were received in an $``init\_values"$ broadcast.
    A nonfaulty $j$ broadcasts its input $x_j\in V$ in its $``init\_values"$ broadcast.
    Since $rec\_vals\subseteq \values{0}{i}$, it contains up to $t$ values broadcasted by faulty nodes and at least $n-2t\geq t+1$ values $x_j\in V$.
    As stated above, $v\in\conv{\Elim{t}{rec\_vals}}$. 
    Using the convexity of the distance function, to complete the proof it is enough to show that there exists a point $x\in V$ such that for every $v_i\in \Elim{t}{rec\_vals}$, $\dist{x}{v_i}\leq \diam{V}$.
    From this point on, the proof is extremely similar to the proof of \cref{lemma::elimIntersection}.
    
    Recall that $\Elim{t}{rec\_vals}$ consists of $t$ iterations of deleting the pair of furthest-distance points.
    For every $i\in[t]$, denote $(p_i,q_i)$ to be the pair deleted in the $i$'th iteration of $\Elim{t}{rec\_vals}$.
    If $\Elim{t}{rec\_vals}\subseteq V$, then for any point $x\in V$ and any point $v_i\in \Elim{t}{rec\_vals}$, $\dist{x}{v_i}\leq \diam{V}$ by the definition of the diameter.
    Otherwise, there exists some $v_i\in\Elim{t}{rec\_vals}$ such that $v_i\notin V$.
    Assume by way of contradiction that there is no $i$ such that both $p_i,q_i\in V$.
    In that case, at least one distinct value from $rec\_vals\setminus V$ is deleted in each iteration.
    There are $t$ such iterations and at most $t$ such values, so this means that all points not in $V$ have been deleted throughout $\Elim{t}{rec\_vals}$.
    In other words, $\Elim{t}{rec\_vals}\subseteq V$, reaching a contradiction.
    Therefore there exists some $i\in[t]$ such that $p_i,q_i\in V$ are the furthest-distance pair in the $i$'th iteration.
    Setting $x=p_i$, we know that for any remaining point $v_i$ at that moment $\dist{v_i}{x}\leq \dist{p_i}{q_i}\leq \diam{V}$.
    All points in $\Elim{t}{rec\_vals}$ must be points that have remained after $i$ iterations, completing the proof.
\end{proof}


\end{document}